\documentclass[11pt]{article}

\usepackage[margin=1in]{geometry}  
\usepackage{graphicx}              
\usepackage{amsmath}               
\usepackage{amsfonts}
\usepackage{dsfont}
\usepackage{amssymb}       
\usepackage{amsthm}                
\usepackage[inline]{enumitem}
\usepackage{color, soul}	
 \usepackage{booktabs}
 \usepackage[table,xcdraw]{xcolor}
\usepackage{hyperref}
\usepackage[titletoc,title]{appendix}
\usepackage{algpseudocode}
\hypersetup{
	colorlinks=true,
	urlcolor=blue,
	linkcolor=blue,
	citecolor=red
	}

\newcolumntype{L}{>{\arraybackslash}m{3cm}}

\usepackage{tikz}
\usetikzlibrary{shapes.geometric, arrows, positioning, backgrounds, fit}
\usepackage{pgfplots}

\pgfplotsset{compat=1.10}
\usepgfplotslibrary{fillbetween}
\usetikzlibrary{patterns}

\newcommand\addvmargin[1]{
	\node[draw=none,fill=none,fit=(current bounding box),inner ysep=#1,inner xsep=0]{};
} 

\newtheorem{thm}{Theorem}
\newtheorem{lem}[thm]{Lemma}
\newtheorem{prop}[thm]{Proposition}
\newtheorem{cor}[thm]{Corollary}

\newtheorem{defn}[thm]{Definition}
\newtheorem{obs}[thm]{Observation}
\newtheorem{clm}[thm]{Claim}

\newtheorem{example}[thm]{Example}

\newcommand{\cl}[1]{\mathcal{#1}} 
\newcommand{\RR}{\mathbb{R}}      
\newcommand{\NN}{\mathbb{N}}      
\newcommand{\EE}{\mathbb{E}}	  

\newcommand{\citet}[1]{\cite{#1}}

\DeclareMathOperator*{\argmin}{argmin}

\DeclareMathOperator*{\supp}{supp}

\makeatletter
\renewcommand*{\@fnsymbol}[1]{\ensuremath{\ifcase#1\or *\or 
		\mathsection\or \mathparagraph\or \|\or **\or \dagger\dagger
		\or \ddagger\ddagger \else\@ctrerr\fi}}
\makeatother

\title{Incentive-Compatible Selection Mechanisms for Forests\thanks{This project has received funding from the European Research Council (ERC) under the European Union's Horizon 2020 research and innovation programme (grant agreement  n$^o  $ 740435).}}

\author{Yakov Babichenko \quad\qquad and \quad\qquad Oren Dean \qquad\quad  and \quad\qquad Moshe Tennenholtz\\
	\small{
		yakovbab@tx.technion.ac.il \qquad\qquad orendean@campus.technion.ac.il \qquad\quad\quad moshet@ie.technion.ac.il\quad\qquad
	}\\
	Technion --- Israel Institute of Technology\\	Haifa, Israel
}

\begin{document}
\maketitle
\begin{abstract}
	Given a directed forest-graph, a probabilistic \emph{selection mechanism} is a probability distribution over the vertex set. A selection mechanism is \emph{incentive-compatible} (IC), if the probability assigned to a vertex does not change when we alter its outgoing edge (or even remove it). The quality of a selection mechanism is the worst-case ratio between the expected progeny under the mechanism's distribution and the maximal progeny in the forest. In this paper we prove an upper bound of 4/5 and  a lower bound of $ 1/\ln16\approx0.36 $ for the quality of any IC selection mechanism. The lower bound is achieved by two novel mechanisms and is a significant improvement to the results of \citet{Babichenko:2018:ID:3178876.3186043}. The first, simpler mechanism, has the nice feature of generating distributions which are fair (i.e., monotone and proportional). The downside of this mechanism is that it is not exact (i.e., the probabilities might sum-up to less than 1). Our second, more involved mechanism, is exact but not fair. We also prove an impossibility for an IC mechanism that is both exact and fair and has a positive quality.
\end{abstract}  
\section{Introduction}
Incentive-compatible selection mechanisms have been studied before in different settings (see Section~\ref{sec: related work}). The motivation for this research stems from many scenarios in which agents approve / disapprove each other, and an administrator is required to select one or more `worthy' agents. To name a few examples:
\begin{enumerate}
	\item The selection of a prize winner in an academic field according to peer reviews.
	\item The selection of an influential user in a social network.
	\item Web search engines select central web-pages by links from other web-pages.
\end{enumerate}
In all of the above examples, the agents have an incentive to misreport their true appreciation or interest in others, if that will lead to their own selection. An \emph{incentive-compatible} (IC) selection mechanism, is a selection mechanism that guarantees that the selection of an agent is independent of his own out-links. A \emph{probabilistic selection mechanism} assigns each agent a selection-probability. The requirement of IC for a probabilistic mechanism is that the probability assigned to an agent will be independent of his own out-links. Selection mechanisms differ in their purpose---some try to maximize a specific graph theoretic measure of the selected agent (most commonly the in-degree) while others guarantee some set of properties (e.g., an agent with a unanimous support is selected, an agent with no support is not selected). \\
In this paper we search for IC, probabilistic selection mechanisms that maximize the expected progeny\footnote{The progeny of a vertex in a network is the number of vertices with paths to this vertex. In a forest, the progeny is the order of the subtree underneath the vertex. See the formal definition in the beginning of Section~\ref{sec: model}.} of the selected agent. We assume that the given network has the structure of a forest; that is, the maximal out-degree is one, and there are no cycles. This network structure arises naturally in cases where the agents join the network sequentially and are allowed to connect to at most one of the users who preceded them. For example, if the agents are customers and a link from customer $ x $ to customer $ y $ denotes that customer $ y $ is the one who referred/recommended the service to customer $ x $, then we get a structure of a forest. In this example, selecting a customer with high progeny means that we are interested in a customer which brought many new customers, both directly (his own sons in the forest) and indirectly (deeper levels in his sub-tree). This example is closely related to mechanisms of multi-level marketing payments (e.g., \citet{Emek:2011:MMM:1993574.1993606, Babaioff:2011:BRB:2325702.2325704, AM-Netecon-11, Douceur:2007:LTM:1282427.1282395}). \\
There are two additional properties of a probabilistic selection mechanism which we consider as desirable. The first is that with probability 1 the mechanism selects an agent (i.e., that the sum of probabilities in any forest is exactly 1). A mechanism with this property is called an \emph{exact} mechanism. The second property is that the mechanism will be \emph{fair}. The requirement of fairness is twofold: \begin{enumerate*}[label=(\alph*)]
	\item monotonicity: higher progeny leads to higher selection-probability;
	\item proportionality: the ratio between the selection-probabilities of two agents depends only on their progenies.
\end{enumerate*} 
See Definition~\ref{dfn: fairness} for a formal definition of these notions. \\
We present two novel IC probabilistic selection mechanisms with a similar worst-case expected progeny-approximation of about 1/3. The first mechanism is fair, but not exact. The second mechanism is exact, but not fair. We then prove an impossibility theorem which states that there are no IC selection mechanisms that are both exact and fair, and with a positive expected progeny-approximation.

\subsection{Related work and our contribution}\label{sec: related work}
Broadly speaking, our paper relates to the track of works on approximate mechanism design without money (\citet{Procaccia:2013:AMD:2542174.2542175}, \citet{Caldarelli07}, \citet{Babaioff:2016:MDS:2872312.2841227}). The goal in these works is to offer mechanisms that provide strategy-proof solutions to problems that have an exact optimal solution which is not strategy-proof (e.g., facility allocation). Naturally, strategyproofness comes at the expense of optimality, and the challenge is to bound this loss.\\
Of those papers who deal with the problem of IC selection in networks, the most similar to ours is \citet{Babichenko:2018:ID:3178876.3186043}. In that paper, the authors offer several IC probabilistic selection mechanisms for trees, forests, and acyclic graphs. They offer a mechanism for forests for which the ratio between the expected progeny and the maximal progeny is proportional to $ n/P^*r $; $ P^* $ being the maximal progeny, $ n $ the number of agents and $ r $ the number of roots in the forest. Clearly this ratio is not bounded from below by any positive number.\footnote{For instance, in a forest with one star of order $ n/2 $ and $ n/2 $ singleton vertices, this yields a ratio of $ 4/n $, which goes to zero as $ n $ goes to infinity.} Hence, our two mechanisms which guarantee a ratio of at least 1/3 {for any forest} are a substantial improvement. \\
The work in \citet{AFPT11} was the first to present the model of IC selection mechanisms in networks, striving to optimize the sum of in-degrees of the selected set. They proved a strong impossibility for an IC deterministic mechanism and offered a probabilistic mechanism based on the idea to randomly partition the agents to voters and candidates. This mechanism does not give a good bound on the progeny of the selected agent. Further works with better mechanisms or slightly different setting can be found in \citet{Fischer:2014:OIS:2600057.2602836, 10.1007/978-3-319-13129-0_10, Bjelde:2017:ISP:3174276.3107922,Kurokawa:2015:IPR:2832249.2832330}.\\
Several works have offered an axiomatic approach to the problem. In these works the authors define a set of desirable axioms and investigate the possibility/impossibility of mechanisms that fulfil maximal subsets of these axioms. To name a few examples of these works,  \citet{ECTA:ECTA1291,MACKENZIE201515,Aziz:2016:SPS:3015812.3015872}. \\
In \cite{Altman:2008:AFR:1622655.1622669} the authors considered the possibility of complete ranking mechanisms under certain axioms.

\emph{Our contribution and paper organization}. In this paper we investigate IC, probabilistic selection mechanisms for forests. We measure the quality of a mechanism as the worst-case ratio between the expected progeny of the selection and the maximal progeny in the forest. The exact model and the formal definitions are in Section~\ref{sec: model}. Maximizing for the progeny is significantly harder than maximizing for the in-degree, since the dependence on a single edge is potentially much larger.
 The mechanism offered in \citet{Babichenko:2018:ID:3178876.3186043} gives a positive bound on the progeny approximation provided the trees in the forest are balanced (the average of their orders is not too far from the maximal order). In Section~\ref{sec: schemes} we suggest two novel mechanisms with a progeny approximation of about 1/3 for all forests. None of these mechanisms is superior to the other; each has a desirable property not present in the other---one of them is fair and the other is exact. 

\section{Preliminaries}\label{sec: model}
A directed forest is an acyclic, directed graph with maximal out-degree 1. Let $ N $ be a set of $ n $ vertices, and let $ F(N,E) $ be a directed forest on $ N $. The roots of $ F $ are those vertices which do not have an out-edge; we denote these vertices by $ R(F) $. For a vertex $ x\in N $, we denote by $ T(x;F) $ the subtree of $ F $ in which $ x $ is the root. We denote by $ P(x;F)=|T(x;F)| $, the \emph{progeny} of $ x $, and by $ P^*(F)=\max\limits_{x\in N}P(x;F)=\max\limits_{x\in R(F)}P(x;F) $, the maximal progeny in $ F $.\\
We will often use the structure of a \emph{star} in our proofs. A $ k $-star is a tree with one {centre} vertex (the root) and $ k-1 $ {leaf} vertices. When we speak of an \emph{isolated} vertex we mean a vertex with no out-edge and no in-edges. \\
Let $ \cl{F}^N $ be the family of all directed forests on $ N $. A selection mechanism is a function $ \cl{M}:N\times\cl{F}^N\to [0,1] $ such that for every $ F\in\cl{F}^N $, $ \sum\limits_{x\in N}\cl{M}(x,F)\leq 1 $. We abuse notation and denote a series of mechanisms $ \{\cl{M}\}_{|N|=1}^{\infty} $ by just $ \cl{M} $. We think of $ \cl{M}(\cdot,F) $ as a probability distribution over $ N $, with the possibility of no-selection. The probability of no selection is $ \cl{M}(\emptyset,F)=1-\sum\limits_{x\in N}\cl{M}(x,F) $. \\
When the forest is obvious from the context, we might omit the extra parameter everywhere and just write: $ R, P^*	,T(x),P(x),\cl{M}(x) $. \\
In this paper we look at mechanisms which are \emph{Incentive-Compatible} (IC). Incentive-compatibility means that for any two forests $ F,F'\in \cl{F}^N $ and a vertex $ x\in N $ such that $ F,F' $ differ only on the out-edge of $ x $, $ \cl{M}(x;F)=\cl{M}(x;F') $.
 Denote by $ F_x $ the forest we get from $ F $ by removing the out-edge of $ x $ (if any).
 \begin{clm}\label{clm: IC equivalence}
Mechanism $ \cl{M} $ is {Incentive-Compatible} if and only if for any forest $ F $ and for any vertex $ x\in N $, $ \cl{M}(x;F)=\cl{M}(x;F_x) $.
 \end{clm}
 \begin{proof}
 	If $ \cl{M} $ is IC, then by definition $ \cl{M}(x;F)=\cl{M}(x;F_x) $  for any $ x,F $. On the other hand, if this condition holds, then for any $ x $ and $ F,F' $ that differ only on the out-edge of $ x $, $ \cl{M}(x;F)=\cl{M}(x;F_x) $ and $ \cl{M}(x;F')=\cl{M}(x;F'_x) $. But $ F_x=F'_x $ (since $ F,F' $ differ only on the out-edge of $ x $), and hence $ \cl{M}(x;F)=\cl{M}(x;F') $, which means that $ \cl{M} $ is IC.
\end{proof}
 A nice consequence of this view of incentive compatibility is that in order to define an IC mechanism, it is enough to define the probabilities assigned to the roots in every forest. 
\begin{cor}\label{cor: roots mecchanism}	
	Let $ \cl{M}_R $ be a mechanism that distributes probabilities on the roots of every forest. Then there is at most one way to extend it to an IC selection mechanism.
\end{cor}
\begin{proof}
	Let $ \cl{M} $ be an IC mechanism that extends $ \cl{M}_R $. For any $ F, x $, if $ x $ is a root then $ \cl{M}(x;F)=\cl{M}_R(x;F) $. Otherwise, by Claim~\ref{clm: IC equivalence}, $ \cl{M}(x;F)=\cl{M}(x;F_x)=\cl{M}_R(x;F_x) $ since $ x $ is a root in $ F_x $. In any case, the probability of $ x $ is determined by $ \cl{M}_R $.
\end{proof}
Not every roots-distribution mechanism is extendible to an IC selection mechanism. For example, the mechanism that always gives the root with highest progeny (with lexicographic tie-breaking) a probability of 0.6 leads to an IC extension which distributes more than 1 already in a forest with a single edge.\footnote{If $ x $ is the first vertex in the lexicographic order, then in the forest with the single edge $ (x,y) $ for any $ y $, both $ x $ and $ y $ get a probability of 0.6.} When we present our two mechanisms in Section~\ref{sec: schemes} we will first define their roots-probabilities and then prove that their IC extensions are well-defined.\\

Of course, there are many IC mechanisms. For instance, the empty mechanism which gives a probability of 0 to all vertices in any forest, or the uniform mechanism which gives a probability of $ 1/n $ to all vertices in any forest. Our goal is to find mechanisms with a good approximation for the maximal progeny, in the worst-case. In this sense, the optimal mechanism is the one which gives a probability of 1 to the vertex with the highest progeny. This mechanism cannot be IC, though. Take for instance the two forests in Figure~\ref{fig: IC example}. The IC requirement implies that $ \cl{M}(a;F_1)=\cl{M}(a;F_2) $. However, in $ F_2 $, $ P(a;F_2)=3=\tfrac{1}{2}P(b;F_2) $. This simple example shows that we cannot avoid distributing at least part of the probability to vertices with progeny at most half of the highest progeny.
\begin{figure}[h]
	\setlength{\tabcolsep}{3mm} 
	\def\arraystretch{1} 
	\centering
	\begin{tabular}{c|c}
		\begin{tikzpicture}[scale=2/3, line width=0.45mm,every node/.style={draw,circle,inner sep=0pt, minimum size=15pt}]
		\node[](a) at (0pt,0pt){$ a $};
		\node[](a1) at (-15pt,-30pt){};
		\node[](a2) at (15pt,-30pt){};

		\node[](b) at (100pt,0pt){$ b $};
		\node[](b1) at (85pt,-30pt){};
		\node[](b2) at (115pt,-30pt){};

		\draw[->] (a1) -- (a);
		\draw[->] (a2) -- (a);
		\draw[->] (b1) -- (b);
		\draw[->] (b2) -- (b);

		\draw(50pt,-50pt) node[below,draw=none,fill=none,inner sep=0pt] {$F_1$};
		\end{tikzpicture}		   &  
		\begin{tikzpicture}[scale=2/3, line width=0.45mm,every node/.style={draw,circle,inner sep=0pt, minimum size=15pt}]
		\node[](a) at (0pt,0pt){$ a $};
		\node[](a1) at (-15pt,-30pt){};
		\node[](a2) at (15pt,-30pt){};
		
		\node[](b) at (100pt,0pt){$ b $};
		\node[](b1) at (85pt,-30pt){};
		\node[](b2) at (115pt,-30pt){};

		\draw[->] (a1) -- (a);
		\draw[->] (a2) -- (a);
		\draw[->] (b1) -- (b);
		\draw[->] (b2) -- (b);
		\draw[->] (a) -- (b);
		
		\draw(50pt,-50pt) node[below,draw=none,fill=none,inner sep=0pt] {$F_2$};
		\end{tikzpicture}		   
	\end{tabular}
\caption{}
\label{fig: IC example}
\end{figure}

We define the quality of mechanism $ \cl{M} $ for the forest $ F $ to be its normalized expected progeny:
\[ Q(\cl{M};F):=\dfrac{\EE[P(x)]_{x\sim\cl{M}(F)}}{P^*(F)}=\dfrac{\sum_{x\in N}\cl{M}(x;F)\cdot P(x;F)}{P^*(F)}. \]
Then the quality of the mechanism is:
\[ Q(\cl{M})=\lim\limits_{|N|\to\infty}\min\limits_{F\in\cl{F}^N}Q(\cl{M};F). \]
In Section~\ref{sec: schemes} we will show two IC mechanisms for which  $ Q(\cl{M})\geq 1/3 $. These are subtle mechanisms, as even coming up with a mechanism with a positive quality is non-trivial.
Clearly, for any mechanism, $ Q(\cl{M})\leq 1 $. In the following proposition we bound it away from 1.

\begin{prop}\label{prp: upper bound}
	For any IC mechanism $ \cl{M} $, $ Q(\cl{M})\leq4/5 $.
\end{prop}
\begin{proof}
	Consider the two forests on two vertices in Figure~\ref{fig: quality bound}.
	\begin{figure}[h]
		\setlength{\tabcolsep}{3mm} 
		\def\arraystretch{1} 
		\centering
		\begin{tabular}{c|c}
			\begin{tikzpicture}[scale=2/3, line width=0.45mm,every node/.style={draw,circle,inner sep=0pt, minimum size=15pt}]
			\node[label={left:$\alpha$}](a) at (0pt,0pt){$ a $};
			\node[label={left:$\alpha$}](b) at (0pt,-40pt){$ b $};
			%
			
			\draw(0pt,-70pt) node[below,draw=none,fill=none,inner sep=0pt] {$F_1$};
			\end{tikzpicture}		   &  
			\begin{tikzpicture}[scale=2/3, line width=0.45mm,every node/.style={draw,circle,inner sep=0pt, minimum size=15pt}]
			\node[label={left:$\beta$}](a) at (0pt,0pt){$ a $};
			\node[label={left:$\alpha$}](b) at (0pt,-40pt){$ b $};
			
			\draw[->] (b) -- (a);
			
			\draw(0pt,-70pt) node[below,draw=none,fill=none,inner sep=0pt] {$F_2$};
			\end{tikzpicture}		   
		\end{tabular}
		\caption{}
		\label{fig: quality bound}
		\end{figure}
		
		To the left of each vertex we have denoted its probability by a generic IC mechanism. Notice that we have used a symmetry assumption when we assumed that in $ F_1 $ the vertices get equal probabilities. In the last paragraph of this section, we explain why symmetry may always be assumed without loss of generality. We calculate the performance of the mechanism for each of these forests.
	\begin{flalign*}
	Q(\cl{M};F_1)&=2\alpha,\\
	Q(\cl{M};F_2)&=\dfrac{2\beta+\alpha}{2}=\beta+\dfrac{\alpha}{2}.
	\end{flalign*}
	The quality of $ \cl{M} $ is at most the minimum of these two expressions. Clearly, we may replace $ \beta $ with its highest possible value, $ 1-\alpha $, to get,
	\begin{flalign*}
	Q(\cl{M})\leq\min\left \{2\alpha,1-\alpha/2 \right \}.
	\end{flalign*}
	Choosing $ \alpha $ which gives the highest minimum, we find that $ \alpha=2/5 $ and $ Q\leq 4/5 $. We remark that we can achieve this bound for any $ n $. We just need to replace the two vertices in Figure~\ref{fig: quality bound} with two $ n/2 $-stars.
\end{proof}
There are two additional properties which we consider as desirable. The first is that the mechanism will be ``fair''. It would be nice if we could require that the probability of a vertex only depends on its progeny. It is not hard to see that this notion of fairness is too strong. For example, start with the empty forest with $ n $ vertices. Then all the vertices get the same probability of $ 1/n $. If we add a single edge from $ x $ to $ y $, then by IC $ x $ still gets $ 1/n $, and by fairness, everyone other than $ y $ should get just like $ x $. This leaves a probability of at most $ 1/n $ for $ y $. Using induction, we can see that for any forest with a single star, the probability of the centre vertex is at most $ 1/n $, which implies an infinitely decaying quality when $ n $ goes to infinity. We suggest the following weaker notion of fairness.
\begin{defn}\label{dfn: fairness}
	Mechanism $ \cl{M} $ is \emph{fair}, if $ \forall F\in\cl{F}^N $ and $ \forall x,y\in R(F) $,
	\begin{enumerate}[label=\alph*)]
		\item (monotonicity) if  $ P(x)> P(y) $, then $ \cl{M}(x;F)\geq\cl{M}(y;F) $; 		
		\item (proportionality) the ratio $ \cl{M}(x;F)/\cl{M}(y;F) $ depends only on $ P(x;F), P(y;F) $.
	\end{enumerate}
\end{defn}
Monotonicity means that the root of a larger sub-tree gets at least the probability of the root of a smaller sub-tree. Proportionality implies in particular that the ratio $ \cl{M}(x)/\cl{M}(y) $ is not influenced by edges outside of $ T(x),T(y) $, nor by the internal structure of $ T(x),T(y) $. In order for the ratio $ \cl{M}(x)/\cl{M}(y) $ to be well-defined, we must also require that the mechanism is positive (i.e., that all the vertices get a positive probability). We can relax this requirement by taking the closure of all positive, fair mechanisms. More precisely, let $ \{\cl{M}_i\}_{i\in\NN} $ be an infinite series of mechanisms. We say that $ \{\cl{M}_i\} $ converges to mechanism $ \cl{M} $ if for every forest $ F\in N $ and every vertex $ x\in N $, $ \cl{M}(x;F)=\lim\limits_{i\to\infty}\cl{M}_i(x;F) $.
We say that $ \cl{M} $ is \emph{fair in limit} if there is a series of positive, fair mechanisms which converge to $ \cl{M} $. If mechanism $ \cl{M} $ is fair in limit, then it can be approximated with a fair mechanism. 

The second desirable feature is that of being exact.
\begin{defn}\label{dfn: budget-balanced}
	Mechanism $ \cl{M} $ is \emph{exact} if for every forest $ F\in\cl{F}^N $, $ \cl{M}(\emptyset,F)=0 $.
\end{defn}
In other words, exactness means that the probabilities assigned by $ \cl{M} $ always sum up to exactly 1. A mechanism which is not exact can only improve its quality, if it will distribute the extra probability; however, the IC requirement might prevent it from doing so.

Let $ x,y\in N $ be such that there is an automorphism of $ F $ that takes $ x $ to $ y $. A \emph{symmetric} mechanism is one for which, under these conditions, $ \cl{M}(x;F)=\cl{M}(y;F) $. Any mechanism can be converted to a symmetric mechanism by picking a random automorphism of $ F $ before using the original mechanism. If the original mechanism was IC/fair/exact and with quality $ Q $, then the symmetric mechanism will possess these features as well. This is why we allowed ourselves to assume in the proof of Proposition~\ref{prp: upper bound} that a general mechanism is symmetric. This observation will also be useful to us in the proof of our impossibility theorem (Section~\ref{sec: imposibility}). Another implication is that our two mechanisms, which are not symmetric, can be made symmetric without hurting any of their properties.

\subsection{Conventions and notations}
As in previous figures, we use circles to denote the vertices in our diagrams. The label of the vertex is marked inside the circle. In the forthcoming diagrams, a $ k $-star is represented by a diamond, with `$ k $' denoted below it. The label inside the diamond is for the centre of the star, and an edge between two diamonds translates to an edge between the centre vertices of the corresponding stars. For example, the tree in Figure~\ref{fig: conventions 1} is equivalent to $ F_2 $ in Figure~\ref{fig: IC example}.
\begin{figure}[h]

\centering
	\begin{tikzpicture}[scale=1, line width=0.45mm,every node/.style={draw,diamond,inner sep=0pt, minimum size=15pt}]

\node[label={below:$3$}](x1) at (0pt,-10pt){$ a $};
\node[label={below:$3$}](x2) at (50pt,-10pt){$ b $};
\draw[->] (x1) -- (x2);

\end{tikzpicture}
	\caption{ }
	\label{fig: conventions 1}
	\end{figure}
	
	When we draw a forest we might add dashed lines and vary the length of the edges so that every vertex is positioned in a distance from the left which is proportional to its progeny. For example, looking at the forest in Figure~\ref{fig: conventions 2} it is clear, both graphically and numerically, that $ P(x_1)<P(y_1)<P(x_2) $ and that $ P(x_3)=P(y_4) $. We emphasize that the dashed lines are not part of the description of the forest and serve only for visualization purposes.
\begin{figure}[h!]
	
	\centering
	\begin{tikzpicture}[scale=1, line width=0.45mm,
	every node/.style={draw,inner sep=0pt, minimum size=15pt}]
	
	\node[diamond,label={below:$4$}](x1) at (80pt,0pt){$ x_1 $};
	\node[diamond,label={below:$2$}](x2) at (120pt,0pt){$ x_2 $};
	\node[diamond,label={below:$ 4 $}](x3) at (200pt,0pt){$ x_3 $};
	
	\draw[dashed] (0pt,0pt) -- (x1);
	\draw[->] (x1) -- (x2);
	\draw[->] (x2) -- (x3);
	
	\node[diamond,label={below:$5$}](y1) at (100pt,-40pt){$ y_1 $};	
	\node[diamond,label={below:$3$}](y2) at (160pt,-40pt){$ y_2 $};		
	\node[circle,label={below:}](y3) at (180pt,-40pt){$ y_3 $};		
	\node[circle,label={below:}](y4) at (200pt,-40pt){$ y_4 $};			

	\draw[dashed] (0pt,-40pt) -- (y1);	
	\draw[->] (y1) -- (y2);
	\draw[->] (y2) -- (y3);
	\draw[->] (y3) -- (y4);
	
	\end{tikzpicture}
	\caption{ }
	\label{fig: conventions 2}
	\end{figure}
	
\section{Two mechanisms}\label{sec: schemes}
In this section we present our two novel mechanisms. To that end we use a few simple observations. Notice that the progeny of any vertex is exactly one more than the sum of progenies of its direct sons. This implies the following.\\

\begin{obs}\label{obs: progeny}\quad
\begin{enumerate}[label=\alph*)]
	\item The progeny of any vertex $ x $ is the highest in its sub-tree, $ T(x) $; the highest progeny in the forest is achieved by at least one of the roots.
	\item The highest progeny in $ T(x)\backslash\{x\} $ is achieved by at least one of the direct sons of $ x $. 
	\item Let $ y,z\in T(x)\backslash\{x\} $ be two different vertices. Assume that $ P(y)\geq\tfrac{1}{2}P(x),P(z)\geq\tfrac{1}{2}P(x) $. Then there must be a path between $ y $ and $ z $.\footnote{Otherwise, there is a vertex $ w\in T(x) $ with a path $ P_y $ from $ y $ to $ w $ and a path $ P_z $ from $ z $ to $ w $ and $ P_y\cap P_z=\{w\} $. This implies that $ P(w)\geq 1+P(y)+P(z)\geq 1+P(x) $, in contradiction to $ a) $.} If the path is from $ y $ to $ z $, then $ P(z)\geq P(y)+1 $, and vice versa.
\end{enumerate}
\end{obs}

We add the notation $ \underline{P}(x)=\max\limits_{y\in T(x)\backslash\{x\}}P(y) $, for the highest progeny in $ T(x) $, excluding $ x $. 

\subsection{A fair mechanism}\label{sec: log scheme}
For our first mechanism we need to assume a total ordering of the vertices by their progenies. We achieve this ordering by breaking ties lexicographically. Hence we will write $ P(x)\succ P(y) \iff (P(x)>P(y))\vee ((P(x)=P(y))\wedge (x<y))  $. Let $ R=\{r_1,r_2,\ldots,r_{|R|}\} $ be a decreasing ordering of the roots by $ \succ $, i.e., $ P(r_i)\succ P(r_{i+1}) $ for all $ i $. We denote by $ r_i(F') $ the $ i $-th root in the forest $ F' $.  Remember that we denote by $ P^*=P(r_1) $ the highest progeny in $ F $.\\
The idea behind our first mechanism is first to recognize the subset of vertices which will get a positive probability (i.e., the support of the mechanism). We need to make sure that this subset is IC and contains only vertices with high progeny. The second step is to set IC probabilities on this subset such that the probabilities are high enough to get a good quality, but not too high as to not distribute more than 1. Specifically, we define the set
\[ A:=\{x\in N: x=r_1(F_x) \}.  \]
That is, $ A $ is the set of vertices $ x $ such that in the graph $ F_x $, $ x $ has the highest progeny (including tie-breaking). Notice that this definition is incentive-compatible in the sense that $ x\in A(F)\iff x\in A(F_x) $, which means that the out-edge of $ x $ does not affect the decision of whether or not it belongs to $ A $.

We use Corollary~\ref{cor: roots mecchanism} and define our mechanism as the IC extension of a roots-distribution mechanism. It is not hard to see that if $ r_1 $ is the only root with a positive probability, we get a mechanism for which the support is precisely the set $ A $ (this can be observed directly from the definition of $ A $). We thus define the following roots-distribution:

\begin{flalign*}
&\cl{M}_f(r_1)=\begin{cases}
\dfrac{1}{2},&|A|=1\\
\dfrac{1}{2}\log_2\dfrac{P(r_1)}{\underline{P}(r_1)},&|A|\geq 2
\end{cases}\\
&\forall r\in R\backslash\{r_1\},\;\cl{M}_f(r)=0,
\end{flalign*}
and we extend it to an IC mechanism which we also denote $ \cl{M}_f $. \\
Before proving the exact properties of $ \cl{M}_f $ (Theorem~\ref{theorem: fair mechanism}), we demonstrate its workings with a couple of examples. Example~\ref{exm: fair 1} is intended to give the intuition that $ \cl{M}_f $ is well-defined; i.e., that $ \sum_{x\in A}\cl{M}_f(x)\leq1 $. Example~\ref{exm: fair 2} is intended to give the intuition that $ \sum_{x\in A}\cl{M}_f(x)\geq \dfrac{1}{2} $. Since $ \forall x\in A, P(x)\geq\tfrac{1}{2}P^* $,\footnote{If $ P(x)<\tfrac{1}{2}P^* $ then in $ F_x $ there is a tree of order $ P^*-P(x)>P(x) $ and $ x $ cannot be the root with the highest progeny.} this readily implies that $ Q(\cl{M}_f)\geq 1/4 $ (though in Theorem~\ref{theorem: fair mechanism} we prove a better bound).
\begin{example}\label{exm: fair 1}
	Consider the forest in Figure~\ref{fig: fair 1}. If we remove the out-edge of $ b $, we get the forest $ F_{b} $, in which $ b $ has the highest progeny. That is, $ b=r_1(F_b) $. In $ F_b $ there is no other vertex which can be the root of the largest tree when we remove its out-edge, hence $ A(F_b)=\{b\} $. Thus, $ \cl{M}_f(b)=\dfrac{1}{2} $. For $ 1\leq i\leq 4 $, $ c_i=r_1(F_{c_i}) $ and $ A(F_{c_i})=\{b,c_1,\ldots,c_i\} $. Hence, $ \cl{M}_f(c_i)=\dfrac{1}{2}\log_2\dfrac{P(c_i)}{\underline{P}(c_i)}=\dfrac{1}{2}\log_2\dfrac{6+i}{6+i-1} $. We now get that 
	\begin{flalign*}
	\sum_{x\in N}\cl{M}_f(x)=\dfrac{1}{2}+\dfrac{1}{2}\sum_{i=1}^{4}\log_2\dfrac{6+i}{6+i-1}=\dfrac{1}{2}+\dfrac{1}{2}\log_2\dfrac{10}{6}<1.
	\end{flalign*}
	\begin{figure}[h!]
		
		\centering
		\begin{tikzpicture}[scale=1, line width=0.45mm,
		every node/.style={draw,inner sep=0pt, minimum size=15pt}]
		
		\node[diamond,label={below:$4$}](x1) at (80pt,0pt){$ a $};
		\node[diamond,label={below:$2$}](x2) at (120pt,0pt){$ b $};
		\node[circle,label={below:}](x3) at (140pt,0pt){$ c_1 $};
		\node[circle,label={below:}](x4) at (160pt,0pt){$ c_2 $};
		\node[circle,label={below:}](x5) at (180pt,0pt){$ c_3 $};
		\node[circle,label={below:}](x6) at (200pt,0pt){$ c_4 $};
		
		\draw[dashed] (0pt,0pt) -- (x1);
		\draw[->] (x1) -- (x2);
		\draw[->] (x2) -- (x3);
		\draw[->] (x3) -- (x4);
		\draw[->] (x4) -- (x5);
		\draw[->] (x5) -- (x6);
		
		\node[diamond,label={below:$5$}](y1) at (100pt,-40pt){$ d $};	
		\draw[dashed] (0pt,-40pt) -- (y1);	
		
		\end{tikzpicture}
		\caption{ }
		\label{fig: fair 1}
	\end{figure}
	
\end{example}
The thing to notice in Example~\ref{exm: fair 1} is that the last vertex in $ A $ gets a probability of at most 1/2, and the sum of probabilities for the rest of the vertices is at most $ \dfrac{1}{2}\log_2\dfrac{P^*}{\tfrac{1}{2}P^*}=\dfrac{1}{2} $. Hence the mechanism is well-defined.
\begin{example}\label{exm: fair 2}
	In Figure~\ref{fig: fair 2}, $ \cl{M}_f(a)=0 $ since in the forest $ F_a $ it is only a second-highest root. In $ F_b $, $ b $ is the highest root and the second highest root is $ d $ with progeny 3. Although vertex $ a $ is not in $ A $, it would be after we remove the out-edge of $ b $. Thus, $ A(F_b)=\{b,a\} $, and $ \cl{M}_f(b)=\dfrac{1}{2}\log_2\dfrac{P(b)}{P(a)}=\dfrac{1}{2}\log_2\dfrac{7}{4} $. Similarly, we find that $ \cl{M}_f(c)=\dfrac{1}{2}\log_2\dfrac{9}{7} $ and $ \cl{M}_f(d)=\dfrac{1}{2}\log_2\dfrac{10}{9} $. Hence,
	\begin{flalign*}
	\sum_{x\in N}\cl{M}_f(x)=\dfrac{1}{2}\log_2\dfrac{10}{4}>\dfrac{1}{2}.
	\end{flalign*}
	\begin{figure}[h!]
		
		\centering
		\begin{tikzpicture}[scale=1, line width=0.45mm,
		every node/.style={draw,inner sep=0pt, minimum size=15pt}]
		
		\node[diamond,label={below:$4$}](x1) at (80pt,0pt){$ a $};
		\node[diamond,label={below:$3$}](x2) at (140pt,0pt){$ b $};
		\node[diamond,label={below:$2$}](x3) at (180pt,0pt){$ c $};		
		\node[circle,label={below:}](x4) at (200pt,0pt){$ d $};
		
		\draw[dashed] (0pt,0pt) -- (x1);
		\draw[->] (x1) -- (x2);
		\draw[->] (x2) -- (x3);
		\draw[->] (x3) -- (x4);		
		
		\end{tikzpicture}
		\caption{ }
		\label{fig: fair 2}
	\end{figure}
	
\end{example}
In Example~\ref{exm: fair 2} no vertex gets the fixed probability of 1/2. This happens because vertex $ b $, which is the last vertex in $ A $, is not the only vertex in $ A(F_b) $. In this case the second vertex in $ A(F_b) $ (i.e., vertex $ a $) has a progeny lower than $ \tfrac{1}{2}P^*(F) $. Hence the total probabilities in this kind of forests is at least $ \dfrac{1}{2}\log_2\dfrac{P^*}{\tfrac{1}{2}P^*}=\dfrac{1}{2} $.\\

Proceeding to a formal analyse of mechanism $ \cl{M}_f $, we first give an alternative definition of the set $ A $ (Claim~\ref{clm: PP mechanism}) and an explicit description of the mechanism's distribution (Lemma~\ref{lem: fair mechanism}). 
\begin{clm}\label{clm: PP mechanism}
	The following is an alternative definition of $ A $:
	\[ 	A=\{x\in T(r_1): P(x)\succ\max\{ \tfrac{1}{2}P(r_1), P(r_2) \}  \}.\footnote{$ P(x)\succ \tfrac{1}{2}P(r_1)\iff ((P(x)>\tfrac{1}{2}P(r_1))\vee((P(x)=\tfrac{1}{2}P(r_1))\wedge(x<r_1)) $.} \]
\end{clm}
\begin{proof}
	If $ x\notin T(r_1) $ then clearly $ x $ is not in the largest tree in $ F_x $, and $ x\notin A $. If $ x\in T(r_1) $ and $ P(x)\prec \tfrac{1}{2}P(r_1) $, then in $ F_x $, $ P(r_1(F);F_x)=P^*-P(x)\succ P(x) $, which means that $ x\neq r_1(F_x) $. Likewise, if $ x\in T(r_1) $ and $ P(x)\prec P(r_2) $ then $ P(r_2(F);F_x)=P(r_2(F);F)\succ P(x;F)=P(x;F_x)  $, and again $ x\neq r_1(F_x) $.\\
	For the other direction, assume that $ x\in T(r_1) $ and $ P(x)\succ \max\{\tfrac{1}{2}P(r_1),P(r_2) \} $. Assume for contradiction that there is a vertex $ y $ such that $ P(y;F_x)\succ P(x;F_x) $. If $ y\notin T(r_1) $ then $ P(y;F_x)\leq P(r_2;F) $ which is a contradiction. If $ y\in T(r_1) $ then by Observation~\ref{obs: progeny} there is a path between $ x $ and $ y $. Since $ P(y;F)\geq P(y;F_x)\succ P(x) $, this path is from $ x $ to $ y $; but then $ P(y;F_x)=P(y)-P(x)\preceq\tfrac{1}{2}P(r_1)\prec P(x) $, which is again a contradiction.
\end{proof}
From Claim~\ref{clm: PP mechanism} and Observation~\ref{obs: progeny} we conclude that $ A $ is a path. Denote $ A=\{r_1=a_{|A|},\ldots,a_1\} $.
\begin{lem}\label{lem: fair mechanism} 
	The support of $ \cl{M}_f(F) $ is $ A(F) $. Furthermore,	\begin{enumerate}[label=\alph*)]
		\item If $ |A|=1 $ then $ \cl{M}_f(r_1)=\dfrac{1}{2} $.
		\item If $ |A|=k\geq 2 $ then
		\begin{flalign*}
		1<\forall i\leq k ,\;\cl{M}_f(a_i)=\dfrac{1}{2}\log_2\dfrac{P(a_i)}{P(a_{i-1})},\\
		\dfrac{1}{2}\log_2\dfrac{2P(a_1)}{P^*}\leq \cl{M}_f(a_1)\leq\dfrac{1}{2}.
		\end{flalign*}
		
	\end{enumerate}	
\end{lem}
\begin{proof}
	The fact that $ \supp(\cl{M}_f)=A $ is immediate from the definitions of $ A $ and $ \cl{M}_f $; and so is claim~$ a) $. \\
	Suppose that $ k\geq 2 $. Since $ A $ is a path, for any $ 2\leq i\leq k $, $ a_{i-1} $ is the vertex with the largest progeny in $ T(a_i)\backslash\{a_i\} $. Clearly this is still true in $ F_{a_i} $, hence $ P(a_{i-1})=\underline{P}(a_i;F_{a_i}) $. From the definition of $ \cl{M}_f $ we get that $ \cl{M}_f(a_i)=\cl{M}_f(a_i;F_{a_i})=\dfrac{1}{2}\log_2\dfrac{P(a_i)}{\underline{P}({a_i};F_{a_i})}=\dfrac{1}{2}\log_2\dfrac{P(a_i)}{P(a_{i-1})} $.\\
	For the lower bound on $ \cl{M}_f(a_1) $, notice that $ \underline{P}(a_1)\prec\max\{\tfrac{1}{2}P(r_1),P(r_2) \} $, otherwise we would have had another vertex in $ A $. If $ \underline{P}(a_1)<\max\{P^*-P(a_1),P(r_2) \} $, then $ A(F_{a_1})=\{a_1\} $, and $ \cl{M}_f(a_1)=\dfrac{1}{2}=\dfrac{1}{2}\log_2\dfrac{2P^*}{P^*}\geq\dfrac{1}{2}\log_2\dfrac{2P(a_1)}{P^*} $.\footnote{This is the case with vertex $ b $ in Example~\ref{exm: fair 1}.} If $ \underline{P}(a_1)\geq\max\{P^*-P(a_k), P(r_2)\}$, then $ \cl{M}_f(a_1)=\dfrac{1}{2}\log_2\dfrac{P(a_1)}{\underline{P}(a_1)}>\dfrac{1}{2}\log_2\dfrac{2P(a_1)}{P^*} $.\footnote{This is the case with vertex $ b $ in Example~\ref{exm: fair 2}.} For the upper bound, notice that if $ |A|\geq 2 $ then $ \underline{P}({r_1})\geq\dfrac{1}{2}P^* $, hence the mechanism never assigns a probability higher than 1/2 to a single vertex.
\end{proof}
We can now prove the main theorem for mechanism $ \cl{M}_f $. The fact that $ \cl{M}_f $ is well-defined is an easy corollary of Lemma~\ref{lem: fair mechanism} and the exact quality is proved by a standard analysis. To prove that this mechanism is fair in limit we build a series of fair mechanisms that converge to our mechanism.
\begin{thm}\label{theorem: fair mechanism}
	Mechanism $\cl{M}_f $ is well-defined, IC, fair (in limit), and with quality $ Q(\cl{M}_f)\geq 1/\ln16\approx 0.36 $.
\end{thm}
\begin{proof}
	The mechanism is IC by definition. If $ |A|=1 $ then only $ r_1 $ has a positive probability of $ \cl{M}_f(r_1)=\dfrac{1}{2} $; and if $ |A|=k\geq2 $,
	\begin{flalign*}
	\sum_{x\in N}\cl{M}_f(x)=\sum_{i=1}^{k}\cl{M}_f(a_i)\leq\dfrac{1}{2}\sum_{i=2}^{k}\log_2\dfrac{P(a_i)}{P(a_{i-
			1})}+\dfrac{1}{2}=\dfrac{1}{2}\log_2\dfrac{P^*}{P(a_{1})}+\dfrac{1}{2}\leq\dfrac{1}{2}\log_2\dfrac{P^*}{\tfrac{1}{2}P^*}+\dfrac{1}{2}= 1, 
	\end{flalign*}
	which shows that this mechanism is well-defined. We turn to bound $ Q(\cl{M}_f) $. If $ |A|=1 $, then clearly $ Q(\cl{M}_f;F)=1/2$. Suppose that $ |A|\geq 2 $. Using Lemma~\ref{lem: fair mechanism} we get,
	\begin{flalign*}
	\EE[P(x)]_{x\sim\cl{M}_f(F)}&> \dfrac{1}{2}\sum_{i=2}^{k}P(a_i)\log_2 \dfrac{P(a_i)}{P(a_{i-1})}+\dfrac{1}{2}P(a_1)\log_2\dfrac{2P(a_1)}{P^*}\\
	&=\dfrac{1}{2\ln2}\sum_{i=2}^{k}P(a_i)\int_{P(a_{i-1})}^{P(a_i)}\dfrac{dz}{z}+\dfrac{1}{2}P(a_1)\log_2\dfrac{2P(a_1)}{P^*}\\
	\geq&\dfrac{1}{2\ln2}\sum_{i=2}^{k}\int_{P(a_{i-1})}^{P(a_i)}dz+\dfrac{1}{2}P(a_1)\log_2\dfrac{2P(a_1)}{P^*}\\
	=&\dfrac{1}{2\ln2}\sum_{i=2}^{k}(P(a_i)-P(a_{i-1}))+\dfrac{1}{2}P(a_1)\log_2\dfrac{2P(a_1)}{P^*}\\
	=&\dfrac{1}{2\ln2}\left  (P^*-P(a_1)+P(a_1)\ln\dfrac{2P(a_1)}{P^*}\right )
	=\dfrac{P^*}{2\ln2}\left(1+\dfrac{P(a_1)}{P^*}\ln\dfrac{2P(a_1)}{eP^*} \right).
	\end{flalign*}
	Since the function $ z\ln(2z/e) $ is monotone increasing in the interval [0.5,1], we get that
	\begin{flalign*}
	\EE[P(x)]_{x\sim\cl{M}_f(F)}\geq \dfrac{P^*}{2\ln2}\left (1+\dfrac{1}{2}\ln\dfrac{1}{e}\right )= \dfrac{P^*}{4\ln2},
	\end{flalign*}
	as claimed. \\
	It remains to show that $ \cl{M}_f $ is fair (i.e., monotone and proportional). For any two vertices $ x,y\in R(F) $ with $ P(x)> P(y) $, it must be that $ \cl{M}_f(y)=0 $, hence it is monotone. To see that it is proportional, define for any $ \epsilon>0 $ the mechanism $ \cl{M}_\epsilon $ induced by the following roots-mechanism:
	\[ \forall r\in R,\; \cl{M}_\epsilon(r)=\cl{M}_f(r_1)\epsilon^{P^*-P(r)}.
	\]
	It is easy to see that $ \cl{M}_\epsilon\to\cl{M}_f $ when $ \epsilon\to  0 $. Notice that for any two roots $ r,r' $, $ \dfrac{\cl{M}_\epsilon(r)}{\cl{M}_\epsilon(r')}=\epsilon^{P(r')-P(r)} $. Since this relation depends only on $ P(r),P(r') $, $ \cl{M}_\epsilon $ is proportional; hence, $ \cl{M}_f $ is fair in limit.
\end{proof}

Both Examples~\ref{exm: fair 1} and \ref{exm: fair 2} show that $ \cl{M}_f $ is not exact. Our next mechanism is an exact mechanism, but not fair.

\subsection{An exact mechanism}\label{sec: BB mechanism}
Consider the forest in Figure~\ref{fig: balanced 1}. Denote $ T_1, T_2, T_3 $ for the largest, second largest and smallest trees, respectively. The vertical dotted line denotes the middle of $ T_1 $ (which is $ P^*/2=5 $). 
	\begin{figure}[h!]
		
		\centering
		\begin{tikzpicture}[scale=1, line width=0.45mm,
		every node/.style={draw,inner sep=0pt, minimum size=15pt}]
		
		\node[diamond,label={below:$2$}](a1) at (40pt,0pt){$ a_0 $};
		\node[diamond,label={below:$4$}](a2) at (120pt,0pt){$ a_1 $};
		\node[diamond,label={below:$ 2 $}](a3) at (160pt,0pt){$ a_2 $};
		\node[diamond,label={below:$ 2 $}](a4) at (200pt,0pt){$ a_3 $};
		
		\draw[dashed] (0pt,0pt) -- (a1);
		\draw[->] (a1) -- (a2);
		\draw[->] (a2) -- (a3);
		\draw[->] (a3) -- (a4);
		
		\node[circle,label={below:}](b1) at (20pt,-40pt){};	
		\node[diamond,label={below:$3$}](b2) at (80pt,-40pt){};	
		\node[diamond,label={below:$3$}](b3) at (140pt,-40pt){$ b_1 $};	
		\node[diamond,label={below:$2$}](b4) at (180pt,-40pt){$ b_2 $};			
		
		\draw[dashed] (0pt,-40pt) -- (b1);
		\draw[->] (b1) -- (b2);
		\draw[->] (b2) -- (b3);
		\draw[->] (b3) -- (b4);
		
		\node[diamond,label={below:8}](c1) at (160pt,-80pt){$ c_1 $};			
		\draw[dashed] (0pt,-80pt) -- (c1);

		\draw[dotted]  (100pt,10pt) -- (100pt,-90pt);
		\end{tikzpicture}
		\caption{ }
		\label{fig: balanced 1}
	\end{figure}

Consider the following mechanism, $ \cl{M}' $, which is exact but not IC. The support of $ \cl{M}' $ are all the vertices which are to the right of the dotted line (namely, $ a_1, a_2, a_3, b_1, b_2, c_1$). The idea is to take the interval $ (\tfrac{1}{2}P^*,P^*]=(5,10] $ and partition it into subintervals. Each vertex $ x\in\supp(\cl{M}') $ gets ownership on the subinterval $ (\max\{\tfrac{1}{2}P^*, \underline{P}(x)\} ,P(x)] $. It is possible that several vertices from different trees will claim ownership of a subinterval, in this case the ownership on this subinterval is equally shared between them. In the example of Figure~\ref{fig: balanced 1}, $ a_1 $ owns the subinterval $ (\tfrac{1}{2}P^*, P(a_1)]=(5,6]  $ in $ T_1 $; $ b_1 $ owns the subinterval $ (5,7] $ in $ T_2 $; and $ c_1 $ owns the subinterval $ (5,8] $ in $ T_3 $. Thus, the subinterval $ (5,6] $ is co-owned by the three of them and each gets a share of 1/3 of this subinterval. Similarly, $ (6,7] $ is shared between $ a_2,b_1,c_1 $; $ (7,8] $ is shared between $ a_2,b_2,c_1 $. The subinterval $ (8,9] $ is shared by only two vertices (because $ P(c_1)\leq 8 $): $ a_3 $ and $ b_2 $. Finally, the subinterval $ (9,10] $ in owned by $ a_3 $ alone. Now each subinteral $ (\alpha,\beta] $ divides a probability of $ \log_2\dfrac{\beta}{\alpha} $ among the partners who own it. We get the following distribution:
\begin{flalign*}
&\cl{M}'(a_1)=\dfrac{1}{3}\log_2\dfrac{6}{5} 
&\cl{M}'(a_2)=\dfrac{1}{3}\log_2\dfrac{7}{6}+\dfrac{1}{3}\log_2\dfrac{8}{7}\\
&\cl{M}'(a_3)=\dfrac{1}{2}\log_2\dfrac{9}{8}+\log_2\dfrac{10}{9}
&\cl{M}'(b_1)=\dfrac{1}{3}\log_2\dfrac{7}{5}\\
&\cl{M}'(b_2)=\dfrac{1}{3}\log_2\dfrac{8}{7}+\dfrac{1}{2}\log_2\dfrac{9}{8}
&\cl{M}'(c_1)=\dfrac{1}{3}\log_2\dfrac{8}{5}
\end{flalign*}
The sum of the probabilities distributed by the subintervals is the probability that would be distributed by the whole interval $ (\tfrac{1}{2}P^*,P^*] $ which is precisely $ \log_2\dfrac{P^*}{\tfrac{1}{2}P^*}=1 $. This shows that this mechanism is exact. The problem, as mentioned, is that it is not IC. To see that, notice that in the forest $ F_{a_1} $, the middle line drops to 4.5 (since now $ P^*(F_{a_1})=P(b_2)=9 $). This means that $ a_1 $ owns a larger subinterval in $ F_{a_1} $, which means higher probability. Similar problem can be for vertices which are in $ T_1 $ right below the middle (but above $ \tfrac{1}{3}P^* $)---removing their out-edge might drop the middle so that they are entitled to some positive probability. If we want our mechanism to be IC we must award them these additional probabilities. The compensation will come from the probability of $ r_1 $ ($ a_3 $ in our example): this vertex will not get his ``fair share'' of the interval but instead will get a probability which completes the total distribution to 1. Of course, we will have to prove that these corrections do not sum up to more than 1 themselves (or equivalently, that the ``complementary probability'' is never negative).\\
We turn to the formal definition of our exact mechanism, $ \cl{M}_b $. For any real number $ z>0 $, we define $ u(z)=|\{r\in R: P(r)\geq z \}| $.
As with the previous mechanism, we define $ \cl{M}_b $ by defining it only for roots:
\begin{flalign*}
\forall i>1,\;\cl{M}_b(r_i)&=\begin{cases}
\dfrac{1}{\ln 2}\int_{\max\{\underline{P}{(r_i)},\tfrac{1}{2}P^* \}}^{P(r_i)}\dfrac{dz}{zu(z)},&P(r_i)>\tfrac{1}{2}P^*,\\
0,&P(r_i)\leq\tfrac{1}{2}P^*.
\end{cases}\\
\cl{M}_b(r_1)&=1-\cl{M}_b(N\backslash\{r_1\}).
\end{flalign*}

We illustrate the workings of this mechanism for the forest in Figure~\ref{fig: balanced 1}. 
\begin{example}\label{exm: balanced 1}
	All the vertices (in all trees) to the right of the middle dotted line get a positive probability. The vertices to the left of this line in trees other than $ T_1 $ get a zero probability. The vertex $ a_0 $ gets a zero probability as well since in the graph $ F_{a_0} $ the middle drops down a little, but it is still above $ P(a_0) $. It is clear from the diagram that
	\begin{flalign*}
	u(z)=\begin{cases}
	3,&z\leq 8,\\
	2, &8<z\leq 9,\\
	1, &9<z\leq 10.
	\end{cases}
	\end{flalign*}
	The vertices $ a_2,b_1,b_2,c_1 $ get the same probabilities as we calculated for the mechanism $ \cl{M}' $.	To find the probability of $ a_1 $, consider the forest $ F_{a_1} $.

Here $ P^*(F_{a_1})=P(b_2)=9 $; hence the middle dropped to $ 4.5 $. We can now calculate,
\[\cl{M}_b(a_1)=\dfrac{1}{3}\log_2\dfrac{P(a_1)}{\tfrac{1}{2}P^*(F_{a_1})}=\dfrac{1}{3}\log_2\dfrac{6}{4.5}. \]
Finally, the probability of $ a_3 $ is the remaining probability which completes the total distribution to 1. Since $ \cl{M}_b,\cl{M}_b' $ are exact and the only differences between them are the probabilities of $ a_1 $ and $ a_3 $, we get that 
\begin{flalign*}
\cl{M}_b(a_3)=\cl{M}'(a_3)-(\cl{M}_b(a_1)-\cl{M}'(a_1))=\dfrac{1}{2}\log_2\dfrac{9}{8}+\log_2\dfrac{10}{9}-\dfrac{1}{3}\log_2\dfrac{6}{4.5}+\dfrac{1}{3}\log_2\dfrac{6}{5}.
\end{flalign*}
\end{example}
If we just wanted to prove that $ \cl{M}_b(F) $ is well-defined, it was enough to show that $ \cl{M}_b(a_1)-\cl{M}_b'(a_1)\leq\cl{M}_b'(a_3) $. Notice also that if the two largest trees were of the same order (i.e., $ P(b_2)=P(a_3) $), then we would have $ P^*(F_{a_1})=P(b_2)=P^*(F) $ and in this case $ \cl{M}_b(a_1)=\cl{M}_b'(a_1) $ (because the middle line, and hence the subinterval of $ a_1 $, does not change when we remove the out-edge of $ a_1 $). This implies that when the two largest trees are of the same order, $\cl{M}_b(r_1)=\cl{M}_b'(r_1) $ (i.e., $ r_1 $ gets his ``fair share''), and all we need to show is that when we lower down the order of the second largest tree, the compensations for the vertices in $ T_1 $ near the middle are never larger than the probability of $ r_1 $ under $ \cl{M}_b' $. The following is our formal claim for the mechanism $ \cl{M}_b $. The proof is based on the above observation. Due to the length and technical nature of the proof, we postpone it to the appendix.
\begin{thm}\label{thm: balanced mechanism}
	Mechanism $ \cl{M}_b $ is well-defined, IC, exact, and with quality $ Q(\cl{M}_b)\geq 1/3$.
\end{thm}

We end this section by showing that $ \cl{M}_b $ is not proportional, and hence not fair. Let $ F $ be a forest with a $ k $-star with centre $ c_1 $ and a $ (k-1) $-star with centre $ c_2 $. Assume that $ n=3k $. The probabilities of $ c_1,c_2 $ are
\begin{flalign*}
\cl{M}_b(c_2)&=\dfrac{1}{2}\log_2\dfrac{k-1}{\tfrac{1}{2}k}=\dfrac{1}{2}(1-\log_2\dfrac{k}{k-1}),\\
\cl{M}_b(c_1)&=1-\cl{M}_b(c_2)=\dfrac{1}{2}(1+\log_2\dfrac{k}{k-1}).
\end{flalign*}
Now, if we add another $ (k-1) $-star (adding such a star does not involve any of the vertices in the trees of $ c_1,c_2 $), then
\begin{flalign*}
\cl{M}_b(c_2)&=\dfrac{1}{3}\log_2\dfrac{k-1}{\tfrac{1}{2}k}=\dfrac{1}{3}(1-\log_2\dfrac{k}{k-1}).\\
\cl{M}_b(c_1)&=1-2\cl{M}_b(c_2)=\dfrac{1}{3}(1+2\log_2\dfrac{k}{k-1}).
\end{flalign*}
We see that the ratio $ \dfrac{\cl{M}_b(c_1)}{\cl{M}_b(c_2)} $ has changed, which means that $ \cl{M}_b $ is not proportional.
\section{An impossibility}\label{sec: imposibility}

In this section we prove the impossibility theorem stated below. 
\begin{thm}\label{thm: Fair impossibility}
	Let $ \cl{M} $ be an IC, fair and exact mechanism. Then $ Q(\cl{M})=0 $.
\end{thm}
Instead of dealing directly with the property of fairness, we prove that this property can be replaced with a more mathematically convenient property of being function-generated (Definition~\ref{dfn: funtion-generated}, Lemma~\ref{lem: fair-FG}). We then prove a couple of asymptotic Lemmata (Lemma~\ref{lem: convexity} and Lemma~\ref{lem: overpaying}) that together can be used to show that for any function-generated mechanism with a positive quality, we can find a forest for which the mechanism is distributing probabilities that sum up to more than one. This is the path with take in the proof of Theorem~\ref{thm: FG impossibility}, which then immediately implies Theorem~\ref{thm: Fair impossibility}. 
\begin{defn}\label{dfn: funtion-generated}
	An exact, IC mechanism $ \cl{M} $ is \emph{function-generated} if there is a series of positive functions $ f_n:\NN\to \RR_+ $ such that for any $ F\in\cl{F}^N $ with $ |N|=n $ and $ r\in R(F) $,
	\[ \cl{M}(r)=\dfrac{f_n(P(r;F))}{\sum_{r'\in R(F)}f_n(P(r';F))}\left (1-\sum_{x\in N\backslash R(F) }\cl{M}(x;F)\right ). \]
\end{defn}
In other words, in a function-generated mechanism, the excess probability (i.e., the probability left after distributing what is due by the IC demand) is linearly distributed between the roots according to $ f_n(P(\cdot)) $. 
\begin{lem}\label{lem: fair-FG}
	Let $ \cl{M} $ be an IC, fair and exact mechanism. Then there is a function-generated mechanism $ \cl{M}' $ such that $ \cl{M}(F)=\cl{M}'(F) $, for every forest $ F $ with at least three roots. 
\end{lem}
\begin{proof}
	Fix $ n $. For any $ 2\leq k\leq n-1 $, let $ S_k $ be the forest with a $ k $-star and $ n-k-1 $ isolated vertices. Let $ c_k $ be the centre vertex of the star and let $ z_k $ be an isolated vertex in $ S_k $. We define the function $ f=f_n $ in the following manner. 
	\begin{flalign*}
	f(1)=1;\\
	\forall 2\leq k\leq n-1, f(k)=\dfrac{\cl{M}(c_k;S_k)}{\cl{M}(z_k;S_k)}.
	\end{flalign*}
	Since $ \cl{M} $ is fair, $ f $ is well-defined. Let $ \cl{M}' $ be the IC mechanism generated by $ f $. We will prove the claim using induction on $ |E(F)| $. The claim is clearly true for the empty forest. Suppose it is true for all forests with $ e-1 $ edges and let $ F $ be a forest with $ e $ edges. Since both $ \cl{M},\cl{M}' $ are IC, the induction hypothesis implies that $ \cl{M}(x)=\cl{M}'(x) $ for any $ x\in N\backslash R $.\footnote{Since $ F_x $ has one less edge and more root than $ F $, the induction applies.} This means that $ \sum_{x\in N\backslash R}\cl{M}(x)=\sum_{x\in N\backslash R}\cl{M}'(x) $, and since both mechanisms are also exact, we get that $ \sum_{x\in R}\cl{M}(x)=\sum_{x\in R}\cl{M}'(x) $. Thus, it is enough to show that $ \forall x,y\in R $, $ \dfrac{\cl{M}(x)}{\cl{M}(y)}=\dfrac{\cl{M}'(x)}{\cl{M}'(y)} $. Denote $ k=P(x), m=P(y) $. Since we assume that there are at least three roots, $ k+m<n $. Let $ S_{k,m} $ be the forest with one $ k $-star, one $ m $-star, and $ n-k-m-2 $ isolated vertices. Let $ c_k, c_m $ be the centre vertices of the $ k $-star and $ m $-star, respectively, and let $ z $ be an isolated vertex. From the proportionality property of $ \cl{M} $ we get
	\begin{flalign*}
	\dfrac{\cl{M}(x;F)}{\cl{M}(y;F)}=\dfrac{\cl{M}(c_k;S_{k,m})}{\cl{M}(c_m;S_{k,m})}=\dfrac{\cl{M}(c_k;S_{k,m})/\cl{M}(z;S_{k,m})}{\cl{M}(c_m;S_{k,m})/\cl{M}(z;S_{k,m})}.
	\end{flalign*}
	Now let $ z_k,z_m $ be isolated nodes in $ S_k,S_m $, respectively. Then again by the proportionality property of $ \cl{M} $,
	\begin{flalign*}
	\dfrac{\cl{M}(c_k;S_{k,m})}{\cl{M}(z;S_{k,m})}=\dfrac{\cl{M}(c_k;S_k)}{\cl{M}(z_k,S_k)}=f(k),\quad\dfrac{\cl{M}(c_m;S_{k,m})}{\cl{M}(z;S_{k,m})}=\dfrac{\cl{M}(c_m;S_m)}{\cl{M}(z_m,S_m)}=f(m).
	\end{flalign*}
	Hence,
	\begin{flalign*}
	\dfrac{\cl{M}(x;F)}{\cl{M}(y;F)}=\dfrac{f(k)}{f(m)}=\dfrac{\cl{M}'(x;F)}{\cl{M}'(y;F)}.
	\end{flalign*}
\end{proof}

In light of Lemma~\ref{lem: fair-FG}, our goal is to show that if $ \cl{M} $ is function-generated, then $ Q(\cl{M})=0 $ (Theorem~\ref{thm: FG impossibility}). To achieve this goal we prove two lemmata. The first, Lemma~\ref{lem: convexity}, states that a function-generated mechanism with a positive quality is ``convex at a distance''; meaning that for any $ x_1,x_2 $ such that $ x_1/x_2 $ is large enough, $ f_n(x_1)/f_n(x_2) $ grows fast with $ n $. 
\begin{lem}\label{lem: convexity}
	Let $ \cl{M} $ be a mechanism generated by the functions $ f_n $. Assume $ Q=Q(\cl{M})>0 $. Then for any $ k,m\in\NN $ such that $ m\geq 2k/Q^2 $, $ f_n(m)=\omega (nf(k)) $.
\end{lem}
\begin{proof}
	Fix $ k $. We will show first that for any $ \ell\geq k\sqrt{2}/Q $, $ f(\ell)=\Omega(nf(k)) $. If we show this, then for any $ m\geq2k/Q^2 $, we can use this claim twice and get that
	\[ f(m)\geq f(2k/Q^2)=\Omega(nf(k\sqrt{2}/Q))=\Omega(n^2f(k))=\omega(nf(k)). \]
	Let $ F $ be the forest on $ n $ vertices with one $ \ell $-star and $ \dfrac{n-\ell}{k} $ $ k $-stars, with $ c_\ell $ as the centre vertex of the $ \ell $-star. Since, 
	\begin{flalign*}
	Q\leq Q(\cl{M};F)\leq\dfrac{\ell\cdot\cl{M}(c_\ell)+k\cdot(1-\cl{M}(c_\ell))}{\ell}\leq\cl{M}(c_\ell)(1-{Q}/{\sqrt{2}})+{Q}/{\sqrt{2}},
	\end{flalign*}	
	we must have $ \cl{M}(c_\ell)=\Omega(1) $. Using $ f $ to bound $ \cl{M}(c_\ell) $ we get
	\begin{flalign*}
	\cl{M}(c_\ell)\leq \dfrac{f(\ell)}{f(\ell)+\tfrac{n-\ell}{k}\cdot f(k)}= \dfrac{1}{1+\tfrac{n-\ell}{k}\cdot \tfrac{f(k)}{f(\ell)}}=\Omega(1)\\
	\Longrightarrow \dfrac{f(k)}{f(\ell)}=O\left (\dfrac{k}{n-\ell}\right )=O(n^{-1}),
	\end{flalign*}
	as needed.
\end{proof}

Lemma~\ref{lem: overpaying}, shows a specific structure of a tree which, under certain conditions, leads to an over-distribution by a function-generated mechanism. 
\begin{lem}\label{lem: overpaying}
	Let $ \cl{M} $ be a mechanism generated by the functions $ f_n $. Let $ F $ be the tree of four connected stars as in Figure~\ref{fig: overpaying lemma}. 

\begin{figure}[h]
	
\centering
	\begin{tikzpicture}[scale=1, line width=0.45mm,every node/.style={draw,diamond,inner sep=2pt}]
	\node[label={below:$b$}](x1) at (0pt,-10pt){$ x_1 $};
	\node[label={below:$b$}](x2) at (50pt,-10pt){$ x_2 $};
	\node[label={below:$a$}](x3) at (100pt,-10pt){$ x_3 $};
	\node[label={below:$a$}](x4) at (150pt,-10pt){$ x_4 $};
	\draw[->] (x1) -- (x2);
	\draw[->] (x2) -- (x3);
	\draw[->] (x3) -- (x4);
	\end{tikzpicture}	
\caption{}\label{fig: overpaying lemma}
\end{figure}
	
Denote $ k=\dfrac{f_n(b)}{f_n(2a)}, m=\dfrac{f_n(a+b)}{f_n(2a)} $. Suppose that $ a=a(n),b=b(n) $ are such that:
	\begin{enumerate*}[label=(\arabic*)]
		\item $ b\geq  2a $,\label{item: b>2a}
		\item $ \dfrac{f_n(1)}{f_n(b)}=o(n^{-1}) $,\label{item: f(1)<f(b)}
		\item $ m\geq 7k^2  $,\label{item: main item}
		\item $\lim\limits_{n\to\infty}\dfrac{f_n(a)}{f_n(2a)}=0 $.\label{item: f(a)<f(2a)}
	\end{enumerate*}

Then $ \sum_{x\in N}\cl{M}(x;F)\geq (1-o(1))\left (1+\dfrac{1}{48k}\right ) $.
\end{lem}
\begin{proof}
	Consider the six sub-forests of $ F $ depicted in Table~\ref{tbl: subforests}. Let $ f=f_n $.
	\begin{table}
		\def\arraystretch{1.25} 
		\centering
		\begin{tabular}{|c|c|c|}
			\hline
			\begin{tikzpicture}[scale=2/3, line width=0.45mm,every node/.style={draw,diamond,inner sep=1pt}]
			\node[label={below:$b$}](x1) at (0pt,-10pt){$ x_1 $};
			\node[label={below:$b$}](x2) at (50pt,-10pt){$ x_2 $};
			\node[label={below:$a$}](x3) at (100pt,-10pt){$ x_3 $};
			\node[label={below:$a$}](x4) at (150pt,-10pt){$ x_4 $};
			\draw(75pt,-60pt) node[below,draw=none,fill=none,inner sep=-10pt] {$F_1$};
			\addvmargin{-50pt}
			\end{tikzpicture}		   &  
			\begin{tikzpicture}[scale=2/3, line width=0.45mm,every node/.style={draw,diamond,inner sep=1pt}]
			\node[label={below:$b$}](x1) at (0pt,-10pt){$ x_1 $};
			\node[label={below:$b$}](x2) at (50pt,-10pt){$ x_2 $};
			\node[label={below:$a$}](x3) at (100pt,-10pt){$ x_3 $};
			\node[label={below:$a$}](x4) at (150pt,-10pt){$ x_4 $};
			\draw[->] (x1) -- (x2);
			\draw(75pt,-60pt) node[below,draw=none,fill=none,inner sep=-10pt] {$F_2$};
			\addvmargin{-50pt}
			\end{tikzpicture}		 &
			\begin{tikzpicture}[scale=2/3, line width=0.45mm,every node/.style={draw,diamond,inner sep=1pt}]
			\node[label={below:$b$}](x1) at (0pt,-10pt){$ x_1 $};
			\node[label={below:$b$}](x2) at (50pt,-10pt){$ x_2 $};
			\node[label={below:$a$}](x3) at (100pt,-10pt){$ x_3 $};
			\node[label={below:$a$}](x4) at (150pt,-10pt){$ x_4 $};
			\draw[->] (x3) -- (x4);
			\draw(75pt,-60pt) node[below,draw=none,fill=none,inner sep=-10pt] {$F_3$};
			\addvmargin{-50pt}
			\end{tikzpicture}		
			\\ \hline
			
			\begin{tikzpicture}[scale=2/3, line width=0.45mm,every node/.style={draw,diamond,inner sep=1pt}]
			\node[label={below:$b$}](x1) at (0pt,-10pt){$ x_1 $};
			\node[label={below:$b$}](x2) at (50pt,-10pt){$ x_2 $};
			\node[label={below:$a$}](x3) at (100pt,-10pt){$ x_3 $};
			\node[label={below:$a$}](x4) at (150pt,-10pt){$ x_4 $};
			\draw[->] (x1) -- (x2);
			\draw[->] (x3) -- (x4);
			\draw(75pt,-60pt) node[below,draw=none,fill=none,inner sep=-10pt] {$F_4$};
			\addvmargin{-50pt}
			\end{tikzpicture}		   &  
			\begin{tikzpicture}[scale=2/3, line width=0.45mm,every node/.style={draw,diamond,inner sep=1pt}]
			\node[label={below:$b$}](x1) at (0pt,-10pt){$ x_1 $};
			\node[label={below:$b$}](x2) at (50pt,-10pt){$ x_2 $};
			+	\node[label={below:$a$}](x3) at (100pt,-10pt){$ x_3 $};
			\node[label={below:$a$}](x4) at (150pt,-10pt){$ x_4 $};
			\draw[->] (x2) -- (x3);
			\draw(75pt,-60pt) node[below,draw=none,fill=none,inner sep=-10pt] {$F_5$};
			\addvmargin{-50pt}
			\end{tikzpicture}		 &
			\begin{tikzpicture}[scale=2/3, line width=0.45mm,every node/.style={draw,diamond,inner sep=1pt}]
			\node[label={below:$b$}](x1) at (0pt,-10pt){$ x_1 $};
			\node[label={below:$b$}](x2) at (50pt,-10pt){$ x_2 $};
			\node[label={below:$a$}](x3) at (100pt,-10pt){$ x_3 $};
			\node[label={below:$a$}](x4) at (150pt,-10pt){$ x_4 $};
			\draw[->] (x1) -- (x2);
			\draw[->] (x2) -- (x3);
			\draw(75pt,-60pt) node[below,draw=none,fill=none,inner sep=-10pt] {$F_6$};
			\addvmargin{-50pt}
			\end{tikzpicture}		
			\\ \hline	
		\end{tabular}
		\caption{}
		\label{tbl: subforests}
	\end{table}
	In each of these forests, all the nodes in  $ N\backslash \{x_1,x_2,x_3,x_4\} $ have progeny 1, and their probabilities are at most $ \tfrac{f(1)}{f(b)}$. Hence by property~\ref{item: f(1)<f(b)}, the total probability of $ N\backslash \{x_1,x_2,x_3,x_4\} $ is at most $ \tfrac{nf(1)}{f(b)}=o(1) $. We will thus ignore the probabilities of these nodes and assume they are actually zero.\\
	We claim that $ f $ must be monotone non-decreasing; otherwise, let $ k>m $ be such that $ f(k)<f(m) $, then
	\begin{flalign*}
	\dfrac{\cl{M}(c_k;S_{k,m})}{\cl{M}(c_m;S_{k,m})}=\dfrac{f(k)}{f(m)}<1,
	\end{flalign*}
	in contradiction to the monotonicity property of a fair mechanism. Thus, since $ 2b>b>2a $, together with property~\ref{item: f(a)<f(2a)}, we get that 
	\begin{flalign*}
	\lim_{n\to\infty}\dfrac{f(a)}{f(b)}=	\lim_{n\to\infty}\dfrac{f(a)}{f(2b)}=0.
	\end{flalign*}
	Hence, in the forest $ F_1 $, $ \cl{M}(x_3;F_1)=\cl{M}(x_4;F_1)=\tfrac{f(a)}{2f(a)+2f(b)}=o(1) $, and similarly in $ F_2 $, $ \cl{M}(x_3;F_2)\leq\tfrac{f(a)}{2f(a)+f(2b)}=o(1) $. Using the IC property we see that $ \cl{M}(x_3;F_3)=\cl{M}(x_3;F_1) $ and we can calculate that $ \cl{M}(x_1;F_3)=\tfrac{f(b)}{f(2a)+2f(b)}(1-\cl{M}(x_3;F_1))=\tfrac{1-o(1)}{2+k^{-1}} $. Again by IC, $ \cl{M}(x_1;F_4)=\cl{M}(x_1;F_3) $, and $ \cl{M}(x_3;F_4)=\cl{M}(x_3;F_2) $. Using IC and $ f(2b)>f(a+b) $, we get
	\begin{flalign}
	\cl{M}(x_2;F)&=\cl{M}(x_2;F_4)=\dfrac{f(2b)}{f(2b)+f(2a)}(1-\cl{M}(x_1;F_3)-\cl{M}(x_3;F_2))\nonumber\\
	&\geq \dfrac{1}{1+m^{-1}}\left (1-\dfrac{1-o(1)}{2+k^{-1}}-o(1)\right )=(1-o(1))\dfrac{1+k^{-1}}{(1+m^{-1})(2+k^{-1})}.\label{eq: x2}
	\end{flalign}
	Below we calculate $ \cl{M}(x_3;F) $ in a similar way without further elaboration.
	\begin{flalign}
	\cl{M}(x_1;F_6)&=\cl{M}(x_1;F_5)=\dfrac{f(b)}{f(a)+f(b)+f(a+b)}(1-\cl{M}(x_2;F_1))\nonumber\\
	&\leq \dfrac{1-1/2(1-o(1))}{1+\tfrac{f(a+b)}{f(b)}}= \dfrac{1+o
		(1)}{2(1+mk^{-1})}.\nonumber\\
	\cl{M}(x_2;F_6)&=\cl{M}(x_2;F_2)= \dfrac{f(2b)}{2f(a)+f(2b)}(1-\cl{M}(x_1;F_1))\leq \dfrac{1}{2}.\nonumber\\
	\cl{M}(x_3;F)&=\cl{M}(x_3;F_6)=\dfrac{f(a+2b)}{f(a)+f(a+2b)}(1-\cl{M}(x_1;F_6)-\cl{M}(x_2;F_6))\nonumber\\
	&\geq (1-o(1))\left (1-\dfrac{1}{2}-\dfrac{1+o(1)}{2(1+mk^{-1})}\right )
	=(1-o(1))\dfrac{mk^{-1}}{2(1+mk^{-1})}.\label{eq: x3}
	\end{flalign}
	Combining (\ref{eq: x2}) and (\ref{eq: x3}) we get:
	\begin{flalign*}
	\sum_{x\in N}\cl{M}(x;F)\geq (1-o(1))\left(\dfrac{1+k^{-1}}{(1+m^{-1})(2+k^{-1})}+\dfrac{mk^{-1}}{2(1+mk^{-1})}  \right).
	\end{flalign*}
	Now, using property~\ref{item: main item} and some algebra,
	\begin{flalign*}
	\dfrac{1+k^{-1}}{(1+m^{-1})(2+k^{-1})}+\dfrac{mk^{-1}}{2(1+mk^{-1})}\geq \dfrac{1+k^{-1}}{(1+k^{-2}/7)(2+k^{-1})}+\dfrac{7k}{2(1+7k)}\\
	=1+\dfrac{1}{2}\cdot\dfrac{35k^3-14k^2-11k-2}{98k^4+63k^3+21k^2+9k+1},
	\end{flalign*}
	and since $ k\geq 1 $ (due to the monotonicity of $ f $), 
	\begin{flalign*}		
	\geq 1+\dfrac{1}{2}\cdot\dfrac{8k^3}{192k^4}=1+\dfrac{1}{48k},
	\end{flalign*}
	as claimed.
\end{proof}

Theorem~\ref{thm: FG impossibility} is just a step away from Theorem~\ref{thm: Fair impossibility}.
\begin{thm}\label{thm: FG impossibility}
	Let $ \cl{M} $ be a function-generated mechanism. Then $ Q(\cl{M})=0 $.
\end{thm}
\begin{proof}
	Assume that $ \cl{M} $ is generated by functions $ f_n $ and has quality $ Q=Q(\cl{M})>0 $. Let $ y $ be the smallest value in $ \{y=2^i/Q^2:i\in \NN, i\geq 1, \dfrac{f_n(2y)}{f_n(y)}>n^{1/\log_2(4/Q^2)}\} $. By Lemma~\ref{lem: convexity}, for large enough $ n $, $ y\leq 4/Q^4 $, for otherwise 
	\[ \dfrac{f_n(4/Q^4)}{f_n(2/Q^2)}\leq\prod_{i=1}^{\lfloor\log_2(4/Q^2)\rfloor}\dfrac{f_n(2^{i+1}/Q^2)}{f_n(2^{i}/Q^2)}\leq \left (n^{1/\log_2(4/Q^2)}\right )^{\lfloor\log_2(4/Q^2)\rfloor}\leq n, \]
	in contradiction to the lemma with $ k=2/Q^2 $.\\
	
	\noindent For two values $ a',b' $, let $ F(a',b') $ be the forest of Lemma~\ref{lem: overpaying} with $ a=a',b=b' $. Consider the set of forests $ \{F_i(y, (2+i)y) \}_{i=0}^{4/Q^2} $. Denote also by $ k_i,m_i $ the values of $ k,m $ in Lemma~\ref{lem: overpaying} for $ F_i $. We claim that for all the forests in this set, we have properties \ref{item: b>2a}, \ref{item: f(1)<f(b)} and \ref{item: f(a)<f(2a)} of the lemma. Indeed, \ref{item: b>2a} is trivial; since $ y\geq2/Q^2 $, Lemma~\ref{lem: convexity} implies property~\ref{item: f(1)<f(b)}; and property~\ref{item: f(a)<f(2a)} comes directly from the definition of $ y $. \\

	\noindent To complete the proof we will show that there is a forest $ F_i $ such that $ m_i\geq 7k_i^2 $ (i.e., property~\ref{item: main item}) and $ k_i $ is uniformly bounded from above. This will allow us to infer that $ \cl{M} $ is over-distributing on $ F_i $.\\
	Notice that $ k_0=1 $, and for any $ i\geq 1 $,
	\begin{flalign*}
	k_i=\dfrac{f_n((2+i)y)}{f_n(2y)};\; m_i=\dfrac{f_n((3+i)y)}{f_n(2y)}=k_{i+1}.
	\end{flalign*}
	Suppose that for all $ 0\leq j\leq i $, $ m_j<7k_j^2 $. Then,
	\begin{flalign}
	m_i=\dfrac{m_i}{k_0}=\prod_{j=0}^i\dfrac{m_j}{k_j}<7^{i+1}\prod_{j=0}^ik_j=7^{i+1}\prod_{j=0}^{i-1}m_j<7^{i+1}\cdot 7^i\prod_{j=0}^{i-1}k^2_j<7^{i+1}\cdot7^i\prod_{j=0}^{i-2}m^2_j\nonumber\\
	<7^{i+1}\cdot 7^{i+2(i-1)}\prod_{j=0}^{i-2}k^4_j<\ldots<7^{i+1}\cdot7^{\sum_{j=1}^{i}j\cdot2^{i-j}}k_0^{2^{i}}=7^{i+1+\sum_{j=1}^{i}j\cdot2^{i-j}}.\label{eq: bound m}
	\end{flalign}
	Hence if $ m_j<7k_j^2 $ for all $ 0\leq j\leq y/Q^2 $, then $ m_{4/Q^2}=\dfrac{f_n((4/Q^2+3)y)}{f_n(2y)} $ is finitely bounded (remember that we bound $ y\leq 4/Q^4 $); but $ \dfrac{(4/Q^2+3)y}{2y}> \dfrac{2}{Q^2} $, which for large enough $ n $ would lead to a contradiction to Lemma~\ref{lem: convexity}.\\
	Now let $ i $ be the first index such that $ m_i\geq 7k_i^2 $. Then we can put $ i=4/Q^4 $ in (\ref{eq: bound m}) to get a uniform upper bound on $ k_i=m_{i-1} $. Hence by Lemma~\ref{lem: overpaying},
	\[ \sum_{x\in N}\cl{M}(x;F)\geq(1-o(1))\left (1+\dfrac{1}{48k_i}\right )>1, \]
	for $ n $ large enough.
\end{proof}
\begin{proof}[Proof of Theorem~\ref{thm: Fair impossibility}]
	Basically, the proof is a direct consequence of Lemma~\ref{lem: fair-FG} and Theorem~\ref{thm: FG impossibility}. However, the claim of Lemma~\ref{lem: fair-FG} applies only to 	 forests with at least three roots, whereas the proof of Lemma~\ref{lem: overpaying} was based on forests with less than three roots. It is not hard, however, to see that the proof will not suffer if we add to $ F $ two isolated vertices, thus the derivation is legitimate.
\end{proof}


\bibliographystyle{siam}
\bibliography{paper}

\appendix

\section{The proof for $ \cl{M}_b $}
\begin{proof}[Proof of Theorem~\ref{thm: balanced mechanism}]
	Mechanism $ \cl{M}_b $ is IC and exact by definition. For any $ x\in N $, if $ P(x)<\tfrac{1}{3}P^* $ then $ P^*(F_x)>\tfrac{2}{3}P^*>2P(x) $ and $ \cl{M}_b(x)=0 $; hence $ \supp(\cl{M}_b)\subseteq\{x\in N: P(x)\geq \tfrac{1}{3}P^* \} $ which implies that $ Q(\cl{M}_b)\geq\tfrac{1}{3} $. By definition $ \cl{M}_b(r_i)\geq 0 $ for all $ i>1 $.  It remains to show that $ \cl{M}_b(r_1)\geq 0 $ for all forests.\\
	Let $ F $ be any forest. It is easy to verify that $ \cl{M}_b(r_1)>0 $ when $ \underline{P}(r_1)\leq \tfrac{1}{2}P^* $. We assume then that $ \underline{P}(r_1)>\tfrac{1}{2}P^* $. For every $ i>1 $ with $ P(r_i)>\tfrac{1}{2}P^* $, $ \supp(\cl{M}_b)\cap T(r_i)=\{x\in T(r_i): P(x)>\tfrac{1}{2}P^* \} $, and since $ \tfrac{1}{2}P^*>\tfrac{1}{2}P(r_i) $, by Observation~\ref{obs: progeny} $ \supp(\cl{M}_b)\cap T(r_i) $ is a path. This means that \[ \cl{M}_b(T(r_i))=\dfrac{1}{\ln2}\sum_{x\in T(r_i):P(x)>\tfrac{1}{2}P^*}\int_{\max\{\underline{P}(x),\tfrac{1}{2}P^* \}}^{P(x)}\dfrac{dz}{zu(z)}=\dfrac{1}{\ln2}\int_{\tfrac{1}{2}P^*}^{P(r_i)}\dfrac{dz}{zu(z)}, \]
	which is independent of the internal structure of $ T(r_i) $. In particular, $ \cl{M}_b(r_1) $ is independent of the internal structure of $ T(r_2) $ and we may assume that $ T(r_2) $ is a $ P(r_2) $-star. For any positive $ p\in \RR $, let $ F(p) $ be the forest we get from $ F $ by replacing the tree $ T(r_2) $ with a $ p $-star.\footnote{This star has $ \lfloor p\rfloor $ leaves and a centre vertex with value $ p-\lfloor p\rfloor $.} If $ P(r_2) $ is infinitesimally close to $ P^* $, then $ \cl{M}_b(r_1)=\dfrac{1}{\ln 2}\int_{\underline{P}(r_1)}^{P(r_1)}\dfrac{dz}{zu(z)} $ (see discussion paragraph after Example~\ref{exm: balanced 1}). To find the probability of $ r_1 $ in $ F $ we start with $ \cl{M}_b(r_1;F(P^*)) $ and integrate the changes in $ \cl{M}_b(r_1;F(p)) $ while lowering down $ p $ until we reach $ P(r_2) $. More precisely, for any $ x\in N,p\in\RR_+ $ we define $ \Delta(x)=\Delta(x,p)=-\dfrac{d}{dp}(\cl{M}_b(x;F(p))) $. By definition of $ \cl{M}_b(r_1) $, $ \Delta(r_1)dp=-\sum_{x\neq r_1}\Delta(x)dp $. Hence we can write,
	\begin{flalign*}
	\cl{M}_b(r_1)&=\cl{M}_b(r_1;F(P(r_2)))= \cl{M}_b(r_1;F(P^*))+\int_{P(r_2)}^{P^*}\Delta(r_1)dp\\	&=\cl{M}_b(r_1;F(P^*))-\sum_{x\neq r_1}\int_{P(r_2)}^{P^*}\Delta(x)dp.
	\end{flalign*}
	We will evaluate the last expression in three intervals, starting with $ p>\underline{P}(r_1) $, continuing with $ \tfrac{1}{2}P^*<p\leq \underline{P}(r_1) $, and finally for $ p\leq\tfrac{1}{2}P^* $. In each interval we will find $ \Delta(x) $ directly by taking an infinitesimal $ \delta>0 $ and denoting $ p'=p-\delta $. Then,
	\[ \Delta(x,p)dp=\cl{M}_b(x;F(p'))-\cl{M}_b(x;F(p)). \]
	Let $ A(p)=\{x\in N:P(x)=P^*(F_x(p)) \}=\{x\in T(r_1): P(x)\geq \max\{p,\tfrac{1}{2}P^*\}\} $. By Observation~\ref{obs: progeny}, $ A $ is a path: $  \{ a_1,\ldots,a_{k}=r_1 \} $. Let $ L(p)=\supp(\cl{M}_b;F(p))\cap T(r_1) $. Let $ \ell $ be the set of leaves in $ L $. Since $ \forall x\in L $, $ P(x)>\tfrac{1}{3}P^* $,  $ |\ell|\leq 2 $.\footnote{If $ L $ has more than two leaves, then there is a ``fork'' in $ T_1 $ with three leaves, each with progeny at least $ \tfrac{1}{3}P^* $. This means that $ P(r_1)\geq P^*+1 $, which is a contradiction.} 
	\begin{enumerate}[label=\arabic*)]
		\item $ p>\underline{P}(r_1) $. In this interval $ A=\{r_1\} $. There are at most three nodes except $ r_1 $ which incur a change in their probabilities under an infinitesimal decrease of $ p $. The first is $ r_2(F(p)) $, whose progeny is decreased, and hence loses part of its probability. Since $ u(p)=2 $ we have,
		\begin{flalign*}
		\Delta(r_2)dp&=\cl{M}_b(r_2;F(p'))-\cl{M}_b(r_2;F(p))=\dfrac{1}{\ln2}\left(\int_{\tfrac{1}{2}P^* }^{p'}\dfrac{dz}{zu(z)}-\int_{\tfrac{1}{2}P^* }^{p}\dfrac{dz}{zu(z)}\right)\\
		&=\dfrac{1}{2\ln2}\int_{p}^{p'}\dfrac{dz}{z}=-\dfrac{1}{2}\log_2\dfrac{p}{p'}. 	
		\end{flalign*}
		The nodes $ x\in\ell $, on the other hand, might gain extra probability due to the decrease in $ P^*(F_x(p)) $ (see the calculation of $ \cl{M}_b(a_1) $ in Example~\ref{exm: balanced 1}). That is, if $ p>P^*-P(x) $ then $ P^*(F_x(p))=p $ and $ P^*(F_x(p'))=p' $. If in addition $ \underline{P}(x)<\tfrac{1}{2}p $, then lowering $ p $ increases the probability of $ x $,
		\begin{flalign*}
		\Delta(x)dp=\dfrac{1}{\ln2}\left(\int_{\tfrac{1}{2}p'}^{P(x)}\dfrac{dz}{zu(z;F_x(p'))}-\int_{\tfrac{1}{2}p}^{P(x)}\dfrac{dz}{zu(z;F_x(p))} \right) \\
		=\dfrac{1}{\ln2}	\int_{\tfrac{1}{2}p'}^{\tfrac{1}{2}p}\dfrac{dz}{zu(z;F_x(p))}=\begin{cases}
		\dfrac{1}{u(\tfrac{1}{2}p)}\log_2\dfrac{p}{p'},&P^*-P(x)<\tfrac{1}{2}p\\
		\dfrac{1}{u(\tfrac{1}{2}p)+1}\log_2\dfrac{p}{p'},&P^*-P(x)\geq\tfrac{1}{2}p.
		\end{cases}
		\end{flalign*}
		The difference between the two cases is that if $ P^*-P(x)\geq\tfrac{1}{2}p $ then in $ F_{x}(p)$ there is another tree with progeny at least $ \tfrac{1}{2}p $ and hence $ u(\tfrac{1}{2}p;F_x(p)) $ increases by one. Consider now the two possibilities:
		\begin{enumerate}
			\item If $ |\ell|=1 $, then we bound $ u(\tfrac{1}{2}p)\geq 2 $ and get that $ \Delta(\ell)dp\leq \dfrac{1}{2}\log_2\dfrac{p}{p'} $. Thus, in this case $ \Delta(r_2)+ \Delta({\ell})\leq 0 $. We conclude that if $ P(r_2)>\underline{P}(r_1) $ and $ |\ell|=1 $, then
			\begin{flalign*}
			\cl{M}_b(r_1)\geq \cl{M}_b(r_1;F(P^*)).
			\end{flalign*}
			\item If $ |\ell|=2 $, then we can bound $ u(\tfrac{1}{2}p)\geq 3 $ and get that $ \Delta(\ell)dp\leq \dfrac{2}{3}\log_2\dfrac{p}{p'} $. We conclude that if $ P(r_2)>\underline{P}(r_1) $ and $ |\ell|=2 $, then $ (\Delta(r_2)+ \Delta({\ell}))dp\leq\dfrac{1}{6}\log_2\dfrac{p}{p'} $ and
			\begin{flalign*}
			\int_{P(r_2)}^{P^*}(\Delta(r_2)+ \Delta({\ell}))dp\leq \dfrac{1}{6}\log_2\dfrac{P^*}{P(r_2)},\\
			\cl{M}_b(r_1)\geq \cl{M}_b(r_1;F(P^*))-\dfrac{1}{6}\log_2\dfrac{P^*}{P(r_2)}.
			\end{flalign*}
		\end{enumerate}
		If $ P(r_2)>\underline{P}(r_1) $ then for $ P(r_2)\leq z\leq P^* $, $ u(z;F(P^*))=2 $ and 
		\[ \cl{M}_b(r_1;F(P^*))=\dfrac{1}{\ln 2}\int_{\underline{P}(r_1)}^{P(r_1)}\dfrac{dz}{zu(z;F(P^*))}\geq \dfrac{1}{2\ln2}\int_{P(r_2)}^{P^*}\dfrac{dz}{z}=\dfrac{1}{2}\log_2\dfrac{P^*}{P(r_2)}. \]
		Thus, from the above two cases we get that if $ P(r_2)>\underline{P}(r_1) $,
		\[ \cl{M}_b(r_1)\geq\begin{cases}
		\dfrac{1}{2}\log_2\dfrac{P^*}{P(r_2)},&|\ell|\leq 1,\\
		\dfrac{1}{3}\log_2\dfrac{P^*}{P(r_2)},&|\ell|=2.
		\end{cases}
		\]
		In any case, $ \cl{M}_b(r_1)>0 $. 
		\item $ \tfrac{1}{2}P^*<p\leq\underline{P}(r_1) $. Taking $ P(r_2)=\underline{P}(r_1) $ in the previous case, we get that
		\begin{flalign*}
		\cl{M}_b(r_1;F(P^*))-\sum_{x\neq r_1}\int_{\underline{P}(r_1)}^{P^*}\Delta(x)dp\geq\begin{cases}
		\dfrac{1}{2}\log_2\dfrac{P^*}{\underline{P}(r_1)},&|\ell|\leq 1,\\
		\dfrac{1}{3}\log_2\dfrac{P^*}{\underline{P}(r_1)},&|\ell|=2.
		\end{cases}
		\end{flalign*}
		Suppose that $ p=\underline{P}(r_1) $ and $ A(p)=\{a_1,a_2=r_1\} $. An infinitesimal decrease of $ p $ affects the probabilities of $ r_2 $ and $ \ell $ just like before. There is another probability which changes with $ p $, namely the probability of $ a_1 $. To evaluate $ \Delta(a_1) $ consider the forest $ F_{a_1} $. Here $ a_1=r_1(F_{a_1}) $ and $ p>\underline{P}(a_1) $. We therefore know from the previous case that $ \Delta(a_1)=-\Delta(r_2(F_{a_1}(p));F_{a_1}(p))-\Delta(\ell(F_{a_1}(p));F_{a_1}(p)) $. Notice that $ \Delta(r_2) $ depends only on $ p $; hence it is the same in $ F(p) $ and $ F_{a_1}(p) $:
		\[\Delta(r_2;F(p))=\Delta(r_2(F_{a_1}(p));F_{a_1}(p)). \]
		We get that \[ \sum_{x\neq r_1}\Delta(x)=\Delta(r_2)+\Delta(\ell)+\Delta(a_1)=\Delta(\ell)-\Delta(\ell(F_{a_1});F_{a_1}). \]
		We claim that $ \Delta(\ell(F_{a_1});F_{a_1})\geq \Delta(\ell) $. To see that, let $ y\in\ell(F_{a_1}) $. Then $ \cl{M}_b(y;F_{a_1})=$\\$\cl{M}_b(y;F_{a_1,y})>0 $.\footnote{The forest $ F_{a_1,y} $ is the forest we get from $ F $ after removing the out-edges of $ a_1 $ and $ y $.} If $ \cl{M}_b(y)=\cl{M}_b(y;F_y)\neq\cl{M}_b(y;F_{a_1,y}) $ then the reason could be one of two: 
		\begin{enumerate}
			\item The first is when $ P(y) $ is above the middle in $ F_{a_1} $ but not in $ F $; that is, $ P^*(F_{a_1,y})=p<2P(y)<P^*(F_y)=P^*-P(y) $. In this case, $ \cl{M}_b(y)=0\Longrightarrow\Delta(y)=0 $. Since $ \Delta(y;F_{a_1})>0 $ (the probability of $ y $ increases when $ p $ drops), we get that $ \Delta(y;F_{a_1})> \Delta(y) $.
			\item The second is when both $ P(y;F)>0 $ and $ P(y;F_{a_1})>0 $, but $ u(\tfrac{1}{2}p;F_{y})=u(\tfrac{1}{2}p;F_{a_1,y})+1 $. This happens when $ P(a_1)-P(y)<\tfrac{1}{2}p $ and $ P^*-P(a_1)<\tfrac{1}{2}p $. Since $ u(\tfrac{1}{2}p;F_{a_1,y})<u(\tfrac{1}{2}p;F_y) $ we get again that $ \Delta(y;F_{a_1})>\Delta(y) $.
		\end{enumerate}
		This analysis remains true for all $ \tfrac{1}{2}P^*<p\leq\underline{P}(r_1) $, no matter the size of $ A $. We conclude that while $ \tfrac{1}{2}P^*<p\leq\underline{P}(r_1) $, $ \sum_{x\neq r_1}\Delta(x)\leq 0 $; hence if $ P(r_2)>\tfrac{1}{2}P^* $ we still have
		\begin{flalign*}
		\cl{M}_b(r_1)\geq\begin{cases}
		\dfrac{1}{2}\log_2\dfrac{P^*}{\underline{P}(r_1)},&|\ell|\leq 1,\\
		\dfrac{1}{3}\log_2\dfrac{P^*}{\underline{P}(r_1)},&|\ell|=2.
		\end{cases}
		\end{flalign*}
		\item $ p<\tfrac{1}{2}P^* $. In this interval $ \Delta(\ell)=0 $ since for $ y\in\ell $, $ P(y)<\tfrac{1}{2}P^* $ and $ P^*(F_y)=P^*-P(y)>\tfrac{1}{2}P^*>p $; hence the decrease in $ p $ will not influence $ \cl{M}_b(y) $. It is possible that for some $ a_i\in A $, $ \Delta(\ell(F_{a_i});F_{a_i})>0 $, but as before, this only works to our advantage.\footnote{That is, a drop in $ p $ would mean higher probability for $ \ell(F_{a_i}) $ which means that $ a_i $ loses probability and then $ r_1 $ gains probability.} In this interval we also have $ \Delta(r_2)=0 $. It might be, though, that $ \Delta(r_2(F_{a_1});F_{a_1})=\Delta(r_2;F_{a_1})<0 $. This happens when $ \underline{P}(a_1)<p<\tfrac{1}{2}P^* $ and $P(a_1)<2p $. In this case $ \Delta(a_1)dp $ might be as high as $ -\Delta(r_2;F_{a_1})dp=\dfrac{1}{2}\log_2\dfrac{p}{p'} $, as we have seen in case $ 1) $. The probability of $ r_1 $ will only be affected while $ p>\tfrac{1}{2}\underline{P}(r_1) $; afterwards the next vertex in $ A $ will compensate for this probability.\footnote{And notice, as before, that $ r_2(a_i) $ is the same for all $ i $ and $ \Delta(r_2) $ only depends on $ p $.} Hence if $ |\ell|= 1 $, then 
		\begin{flalign*}
		\cl{M}_b(r_1)\geq \dfrac{1}{2}\log_2\dfrac{P^*}{\underline{P}(r_1)}-\int_{\tfrac{1}{2}\underline{P}(r_1)}^{\tfrac{1}{2}P^*}\Delta(a_1)dp\geq \dfrac{1}{2}\log_2\dfrac{P^*}{\underline{P}(r_1)}
		-\dfrac{1}{2}\log_2\dfrac{P^*}{\underline{P}(r_1)}=0.
		\end{flalign*}
		To complete the proof we will show that if $ |A|\geq 2 $ then while $ p>\underline{P}(r_1) $, $ |\ell|= 1 $. Indeed assume the opposite. Let $ x=\argmin\limits_{x\in\ell}P(x) $. Then $ p>\underline{P}(r_1)\geq P(a_1)>2P(x) $. However, then $ P(x)<\dfrac{1}{2}p\leq\dfrac{1}{2}P^*(F_x(p)) $ and $ \cl{M}_b(x)=0 $, in contradiction. 
	\end{enumerate}
\end{proof}

\end{document}